\documentclass{article}
\usepackage{amsmath}
\usepackage{amssymb}
\usepackage{amsthm}
\newtheorem{mydef}{Definition}
\newtheorem{ex}{Example}

\newtheorem{thm}{Theorem}

\newtheorem{lem}{Lemma}
\newtheorem{rem}{Remark}
\newtheorem{cor}{Corollary}
\newtheorem{prop}{Proposition}
\newtheorem{alg}{Algorithm}
\usepackage[margin=1in]{geometry}
\usepackage{verbatim}

\newcommand{\N}{\mathbb{N}}
\newcommand{\Z}{\mathbb{Z}}

\newcommand{\F}{\mathbb{F}}
\newcommand{\ra}{\xrightarrow{}}

\newcommand{\mc}[1]{\mathcal{#1}}

\newcommand{\abs}[1]{\left| #1 \right|}
\newcommand{\pren}[1]{\left( #1 \right)}

\title{Error Correcting Codes}
\author{Priyank Deshpande }

\begin{document}

\maketitle

\textbf{Abstract} 
Here we present some revised arguments to a randomized algorithm proposed by Sudan to find the polynomials of bounded degree agreeing on a dense fraction of a set of points in $\F^{2}$ for some field $\F$.

\section{Introduction}
$\text{Here we will discuss some concepts in the field of error-correcting codes. }$

\begin{mydef}
    Given $\Sigma$ a collection of symbols, and $x,y \in \Sigma^{n}$. We define the hamming distance between $x$ and $y$ denoted as $HD(x,y)$ as $\abs{\{ i \in [n] : (x)_{i} \ne (y)_{i} \}}$. That is, the number of indices at which $x$ and $y$ differ. 
\end{mydef}

\begin{ex}
    Given $\Sigma = \{0,1,2\}$ and $x = ``201"$, and $y = ``222"$. We have that $HD(x,y) = 1$.
\end{ex}

\begin{mydef}
    Let $\Sigma$ be a collection of symbols and $n, k, \delta \in \Z$. We say $\mc{C} \subset \Sigma^{n}$ is a $[n,k,\delta]$ code if $\abs{\mc{C}} = \abs{\Sigma}^{k}$, and $\forall x,y \in \mc{C}, HD(x,y) \geq \delta$.
\end{mydef}

\begin{mydef}
    Let $\Sigma$ be a collection of symbols and $\mc{C}$ a $[n,k,\delta]$. If $\tau \in \Z : 2 \tau + 1 \leq \delta$, then we say $\mc{C}$ is a $\tau$ error correcting code. 
\end{mydef}

\begin{mydef}
    Let $F$ be a finite field of cardinailty $n$. Let $\mc{C}$ be a $[n,d+1,n-d]$ code of an alphabet $\Sigma$, we say $\mc{C}$ is a Reed-Solomon Code if $\mc{C} = \{ ``p(0) | p(w) | \dots | p(w^{|F|-1})" :  p(x) \in F[x], \deg(p) \leq d \}$. Here $|$ denotes string concatenation, and $w \in F$, is a generator of $F^{*}$. 
\end{mydef}

\begin{mydef}
    We will refer the maximum-likelihood decoding problem as the following task: Given a $[n,k,\delta]$ code, a string $s \in \Sigma^{n}$, a string in $c \in \mc{C}$ such that $HD(s,c) \leq HD(s,x) , \forall x \in \mc{C}$. We will refer to the list decoding problem as: Given a string $s\in \Sigma^{n}$, a $[n,k,\delta]$ code, and a parameter $\tau \in \N$, return all $c \in \mc{C}$ such that $HD(s,c) \leq \tau$,
\end{mydef}

\section{Algorithm}

\begin{rem}
    We will present a randomized algorithm by Sudan-'96, for the following problem: Given a field $F$, $\{(x_{i},y_{i})\}_{i=1}^{n} \subset F^{2}$ and parameters $t,d \in \N$, find all $f(x) \in F[x]$ such that $\abs{\{ i \in [n] : f(x_{i}) = y_{i})} \geq t$, and $\deg(f) \leq d$. 
\end{rem}

\begin{rem}
    We define concept of weighted degree which will be relevant to the randomized algorithm to be presented. Given $(w_{x} , w_{y}) \in \Z^{2}$ which we will call weights and a bivariate monomial in $x,y$, $c_{ij}x^{i}y^{j}$, we say the weighted degree of such a monomial is $i \cdot w_{x} + j \cdot w_{y}$. Given a bivariate polynomial $Q(x,y) \in F[x,y]$, we say the weighted degree of $Q(x,y) = \sum_{i,j} c_{ij} x^{i} y^{j}$ to be the maximum of the weighted degrees of its monomials.  
\end{rem}

\begin{alg}
    Define the following randomized algorithm: Let $\{(x_{i},y_{i})\}_{i \in [n]}$, $d,t \in \N$ be inputs to the algorithm, and $m,l$ parameters to be determined to optimize the algorithm. Then:
        \begin{itemize}
            \item Find a $P(x,y) \in F[x,y]$ such that $P(x,y)$ has weighted degree with weights $(1,d)$ at most $m+l \cdot d$, $P(x,y)$ is not identically zero, and $P(x,y)$ vanishes on $\{(x_{i},y_{i})\}_{i \in [n]}$. That is, $P(x_{i}, y_{i}) = 0, \forall i \in [n]$. $(1)$
            \item Factor $P(x,y)$ into irreducible polynomials in $F[x,y]$. $(2)$
            \item Check all functions $f(x) \in F[x]$ of degree at most $d$, such that $(y - f(x)) \mid P(x,y)$, and $f(x_{i}) = y_{i}$ for at least $t$ distinct choices of $i \in [n]$. $(3)$
        \end{itemize}
\end{alg}

\begin{rem}
    We will justify that this algorithm runs in polynomial time. 
\end{rem}

\begin{prop}
    The polynomial as described in step $(1)$ can be found in polynomial time, with respect to the size of the field, if such a polynomial exists.
\end{prop}

\begin{proof}
    By the conditions imposed by the weighted degree constraint, we can write $P(x,y) \in F[x,y]$ as $P(x,y) = \sum_{j=0}^{l} \sum_{i}^{m+(l-j)d} c_{ij} x^{i} y^{j}$ because $j \leq l$, $i \leq m + (l-j)d$ implies that $(i,j) \cdot (1,d) \leq (m+(l-j)d,j) \cdot (1,d) = m + ld$, which is the weighted degree of $P(x,y)$. To find the polynomial which satifies the conditions in $(1)$, we require that $\sum_{j=0}^{l} \sum_{i=0}^{m+(l-j)d} c_{ij} (x_{k})^{i} (y_{k})^{j} = 0, \forall k \in [n]$. Let $|F| = N$. Using a brute force approach to determine the appropriate values of $c_{ij}$, we can obtain a solution in $O ( n \cdot N ^{(m+ld)l} )$, which is polynomial in $N$ for fixed parameters $l,m,d$. However, this can be solved in polynomial time with respect to the number of constraints $n$. 
\end{proof}

\begin{prop}
    If the parameters $m,l$ are such that $(m+1)(l+1) + d\binom{l+1}{2} > n$, then a function $P(x,y) \in F[x,y]$ as described in $(1)$ exists.
\end{prop}

\begin{proof}
    Let $\eta = (m+1)(l+1) + d\binom{l+1}{2}$ Note that if $P(x,y) \in F[x,y]$ is defined as $P(x,y) = \sum_{j=0}^{l} \sum_{i=0}^{m+(l-j)d} c_{ij} x^{i} y^{j}$, then there are $\eta$ many $x_{ij}$'s. To find the polynomial $P(x,y)$, we need to solve the system $A \Vec{x} = 0$, where $\Vec{x}$ represents the $c_{ij}$, and $A$ has dimensions $n \times \eta$. So this amounts to finding the null space of $A$. Under the assumption that $\eta > n$, we have that $\dim(N(A)) \geq 1$, where $N(A)$ is the null space of $A$. Hence, we may choose a $y \in N(A) \setminus \{0\}$ to obtain the desired $c_{ij}$'s.
\end{proof}

\begin{prop}
    If $P(x,y) \in F[x,y]$ satisfies (1), and $f(x) \in F[x]$ satisfies $\abs{\{i \in [n] : f(x_{i}) = y_{i} \}} \geq t$, and $t > m + ld$, then $y-f(x)$ divides $P(x,y)$. 
\end{prop}

\begin{rem}
    Let $f(x) \in F[x]$. Denote the condition $\abs{\{i \in [n] : f(x_{i}) = y_{i} \}} \geq t$ as (*), and say $f(x)$ satifies (*) should it be the case
\end{rem}

\begin{proof}
    Let $f(x) \in F[x]$ satisfy (*). We claim that $P(x,f(x))$ is identically zero. Since $P(x,y)$ has $(1,d)$ weighted degree at most $m + ld$, we have that $P(x,f(x))$ (as a uni-variate polynomial) has degree at most $m+ld$ since $f(x)$ has degree at most $d$. However $P(x,f(x)) = 0$ whenever $x = x_{i}$ for some $i \in [n]$. If $f(x)$ satisfies (*), then there are at least $t$ zeros. Under the assumption that $t > m + ld$, we have that the number of roots of $P(x,f(x))$ is greater than its degree, so $P(x,f(x)) \equiv 0$. Consider $P(x,y) = P_{x}(y) = \sum_{j=0}^{l-1} P_{j}(x) y^{j} : P_{j}(x) \in F[x]$. Since $P_{x}(f(x)) = 0$, we have that $P_{x}(y)$ has a root $f(x)$. By the division algorithm, $(y-f(x))$ divides $P_{x}(y) = P(x,y)$, which is the claim. 
\end{proof}

\begin{rem}
    It remains to choose the parameters $m,l$ such that $t > m + ld$ and $(m+1)(l+1) + d \binom{l+1}{2} > n$, We can rephrase the condition to be $(m+1)(l+1) + d \binom{l+1}{2} \geq n + 1$ Observe that this condition yields $m \geq \frac{n + 1 - d \binom{l+1}{2}}{l+1} - 1$. Suppose that we want $t \geq m + ld + 1 \implies t \geq \frac{n + 1 - d \binom{l+1}{2}}{l+1} + ld = \frac{n+1}{l+1} - \frac{d l}{2} + dl = \frac{n+1}{l+1} + \frac{dl}{2}$. To find the minimum of this function with respect to $l$, we perform a first derivative which yields that $\frac{-(n+1)}{(l+1)^{2}} + \frac{d}{2} = 0 \implies l = \sqrt{\frac{2(n+1)}{d}} - 1$. Substituting the expression for $l$ in for the expression on $m$, we obtain that $m \geq \frac{n + 1 - d \binom{l+1}{2}}{l+1} - 1 = \frac{n+1}{l+1} - \frac{dl}{2} - 1 = \sqrt{\frac{d(n+1)}{2}} - \left( \sqrt{\frac{d(n+1)}{2}} - \frac{d}{2} \right) - 1 = \frac{d}{2} - 1$. This yields for the condition on $t$ that $t \geq m + ld + 1\geq \frac{d}{2} + d \cdot \left( \sqrt{\frac{2(n+1)}{d}} - 1 \right) = \frac{d}{2} + \sqrt{2(n+1)d} - d = \sqrt{2(n+1)d} - \frac{d}{2}$. This will allow us to make the following claim, which follows from the previous propositions. 
\end{rem}

\begin{cor}
    Given a field $F$ and a set of points $\{ (x_{i}, y_{i}) \}_{i \in [n]} \subset F^{2}$, and paramters $d,t \in \N$ such that $t \geq d \cdot \lceil \sqrt{\frac{2(n+1)}{d}} \rceil - \lfloor \frac{d}{2} \rfloor$, then there is a polynomial time algorithm in $n$ which finds all polynomials $f(x) \in F[x]$ which satisfy (*), and have degree at most $d$.
\end{cor}

\begin{proof}
    Setting $m = \lfloor \frac{d}{2} \rfloor - 1$ gives and $l = \lceil \sqrt{\frac{2(n+1)}{d}} \rceil - 1$ gives that $(m+1)(l+1) + d \binom{l+1}{2} \geq n + 1$. By proposition 2, a function $P(x,y)$ not identically zero satisfying that $P(x_{i}, y_{i}) = 0, \forall i \in [n]$ exists. Under the assumption that $t \geq d \cdot \lceil \sqrt{\frac{2(n+1)}{d}} \rceil - \lfloor \frac{d}{2} \rfloor > m + ld$, we have that $(y-f(x))$ divides $P(x,y)$ should such an $f(x) \in F[x]$ satisfy (*). By step 3 in the algorithm, $f(x)$ will be reported as output.
\end{proof}

\begin{mydef}
    Denote the tuple of $k$ variables $(x_{1}, \dots , x_{k}) = \Vec{x}$. Let $F$ be a field, $H \subset F$, and $g : H^{k} \ra F$. Given parameters $t,d \in \N$, output all polynomials $f$ of degree at most $d$ such that $\abs{\{ \Vec{x} \in H^{k} : f(\Vec{x}) = g(\Vec{x})\}} \geq t$. Here define the degree of $f$ to be the maximum degree of its monomials. 
\end{mydef}

\begin{rem}
    Let $f(x) \in F[\Vec{x}]$. We say $f(x)$ satisfies the condition (*) if $\abs{ \{ \Vec{x} \in H^{k} : f(\Vec{x}) = g(\Vec{x}) \} } \geq t$.
\end{rem}

\begin{mydef}
    Generalize the definition of weight degree to an $n$-variate polynomial. First, the $(w_{1}, \dots , w_{n})$ weighted degree of a monomial $\prod_{i=1}^{n} x_{i}^{d_{i}}$ is defined to be $\sum_{i=1}^{n} w_{i} d_{i}$. Define the $(w_{1}, \dots , w_{n})$ weighted degree of an $n$-variate polynomial to be the maximum of the weighted degrees of its monomials (which have non-zero coefficients).  
\end{mydef}

\begin{alg}
    Define the following algorithm. Let $F,H,k,t,d,g$ be as in definition $6$, and $m,l \in \N$ be parameters to be determined. 
    \begin{itemize}
            \item Find a $P(x_{1}, \dots, x_{k} ,y) \in F[x_{1}, \dots , x_{k}, y]$ such that $P(\Vec{x},y)$ has weighted degree with weights $(1,\dots , 1 ,d)$ at most $m+l \cdot d$, $P(\Vec{x},y)$ is not identically zero, and $P(\Vec{x},y)$ vanishes on $\{ (\Vec{x}, g(\Vec{x})) : \Vec{x} \in H^{k} \}$. That is, $P(\Vec{x}, g(\Vec{x})) = 0, \forall \Vec{x} \in H^{k}$.
            \item Factor $P(\Vec{x},y)$ into irreducible polynomials in $F[\Vec{x},y]$. $(2)$
            \item Check all functions $f(\Vec{x}) \in F[\Vec{x}]$ of degree at most $d$, such that $(y - f(\Vec{x})) \mid P(\Vec{x},y)$, and $f(\Vec{x}) = g(\Vec{x})$ for at least $t$ distinct choices of $\Vec{x} \in H^{k}$. $(3)$
        \end{itemize}
\end{alg}

\begin{rem}
    This is more or less a generalization of the uni-variate case. We will state conditions for the existence of $P(\Vec{x},y)$, and show that a polynomial $f(\Vec{x}) \in F[\vec{x}]$ satisfying (*) will be such that $y - f(\vec{x}) \mid P(\vec{x},y)$. Let $|H| = h$.
\end{rem}

\begin{prop}
    If $m+ld \geq k(h-1)$, then a non-trivial polynomial $P(\vec{x},y) \in F[\vec{x},y]$ vanishing on $S = \{ (\vec{x},f(\vec{x})) \in F^{k+1} : \vec{x} \in H^{k}\}$ exists. 
\end{prop}

\begin{proof}
        We want to show that the number of monomials of a $k+1$ variate weighted degree polynomial is greater than $|H|^{k}$. Then we can apply a similar argument for there being a non-trivial solution to the system of linear equations $A\vec{z} = \vec{0}$ for finding the coefficients of such a polynomial $P(\vec{x},y)$. We observe that a polynomial of $(1,\dots,1,d)$ weighted degree $m+ld$ contains $\sum_{j=0}^{l} \binom{m+(l-j)d + k}{k}$ monomials. This is because $P(\vec{x},y) = \sum_{j=0}^{l} P_{j}(\vec{x})y^{j}$, where $P_{j}$ has total degree at most $m+ld-jd = m + (l-j)d$. Hence, let $M(Q)$ denote the number of distinct monomials of a polynomial $Q$. Then $M(P) = \sum_{j=0}^{l} M(P_{j}) = \sum_{j=0}^{l} \binom{m+(l-j)d + k}{k}$, which implies the claim. Now we would like $M(P) > h^{k}$. To do this, we provide some lower bounds. Observe that $M(P) = \sum_{j=0}^{l} \binom{m+(l-j)d + k}{k} \geq \sum_{j=0}^{l} \left( \frac{m + (l-j)d + k}{k} \right) ^{k} \geq \pren{\frac{m+ld+k}{k}}^{k} + l \cdot \pren{\frac{m+k}{k}}^{k} > \pren{\frac{m+ld+k}{k}}^{k} \geq \pren{\frac{k(h-1)+k}{k}}^{k} = h^{k}$, which proves the proposition.   
\end{proof}

\begin{prop}
    If $t > (m+ld)h^{k-1}$, where $t$ is the number of agreements of a $k$-variate polynomial $f$ on the set $S = \{ (\vec{x},g(\vec{x})) \in F^{k+1} : \vec{x} \in H^{k} \}$, then $y - f(\vec{x}) \mid P(\vec{x},y)$. 
\end{prop}

\begin{proof}
    We observe that $\theta_{f}(\vec{x}) = P(\vec{x},f(\vec{x}))$ is a $k$-variate polynomial of total degree $m+ld$. Let $Z(Q,S) = \{ x \in S : Q(x) = 0\}$. By the Schwartz-Zippel Lemma, if $\abs{Z(\theta_{f},H^{k})} > \deg(\theta_{f}) \cdot |H|^{k-1}$, then $\theta_{f} = P(\vec{x},f(\vec{x})) \equiv 0$. But $t = \abs{Z(\theta_{f},H^{k})} > \deg(\theta_{f}) \cdot |H|^{k-1} = (m+ld)h^{k-1}$ by assumption so $P(\vec{x},f(\vec{x})) \equiv 0$, and $y - f(\vec{x}) \mid P(\vec{x},y)$, since it is a root of $P(\vec{x},y)$.
\end{proof}

\begin{lem}
    (Schwartz-Zippel) Let $p(x_{1},\dots , x_{n}) \in F[x_{1},\dots , x_{n}]$ be a polynomial of total degree $d$ that is not equivalently $0$. Let $|S| \subset F$ be an arbitrary finite subset of the field $F$. Then $\mathbb{P}_{\vec{x}\in_{R} S^{k}}[p(\vec{x}) = 0] \leq \frac{d}{|S|}$.
\end{lem}

\begin{proof}
    We proceed by induction. Considering the uni-variate case yields that $\mathbb{P}_{x \in_{R} S} [p(x) = 0] \leq \frac{d}{|S|}$. This is true because $p$ has degree $d$ and since $p$ is not identically $0$ we have that there are at most $d$ roots in $F$, and hence there are at most $d$ roots of $p$ in a finite subset $S \subset F$. Let $k = \deg_{x_{n}}(p)$. Then we may write $p(x) = x_{n}^{k} q(x_{1},\dots,x_{n-1}) + r(x_{1}, \dots, x_{n})$, where $q(x_{1},\dots,x_{n-1})$ has total degree at most $d-k$ and $r(x_{1},\dots,x_{n})$ has $x_{n}$ degree strictly less than $k$. For a $\vec{x} \in_{R} S^{k}$, we have that 
    $$
    \mathbb{P}[p(\vec{x}) = 0] = \mathbb{P}[p(\vec{x}) = 0 \mid q(\vec{x}) \ne 0] \mathbb{P}[q(\vec{x}) \ne 0] + \mathbb{P}[p(\vec{x}) = 0 \mid q(\vec{x}) = 0] \mathbb{P}[q(\vec{x}) = 0]
    $$
    by Bayes Formula. But $\mathbb{P}[p = 0] \leq \mathbb{P}[q = 0] + \mathbb{P}[p = 0 \mid q \ne 0] \leq \frac{d-k}{|S|} + \frac{k}{|S|} = \frac{d}{|S|}$, by the inductive hypothesis. This completes the proof.  
\end{proof}

\begin{rem}
    There is some subtlety to the fact that $\mathbb{P}[p = 0 \mid q \ne 0] \leq \frac{k}{|S|}$. This is because if we are given that $q \ne 0$, then $p(x_{1},x_{2}, \dots , x_{n})$ considered in that regard becomes a uni-variate polynomial of degree $k$ in the variable $x_{n}$ that is not identically $0$, as we are fixing that $x_{1},\dots,x_{n-1}$, and we then may apply the inductive hypothesis. The application of the Schwartz-Zippel Lemma as presented here to proposition $5$ is that the contrapositive of the statement of lemma $1$ is sufficient to deduce that $\theta_{f} \equiv 0$ as it has more than $\deg(f) \cdot h^{k-1}$ roots on $H^{k}$.
\end{rem}

\begin{thm}
    If the parameters $d,t,|H| = h, k \in \N$, are such that $\frac{t}{dh^{k-1}} > \frac{k(h-1)}{d}$, and the open interval $\pren{\frac{k(h-1)}{d},\frac{t}{dh^{k-1}}}$ contains a positive integer, then Algorithm 2 as described above outputs all the desired polynomials. 
\end{thm}

\begin{proof}
    Note that to obtain the non-trivial polynomial $P(\vec{x},y) \in F[x_{1},\dots, x_{k},y]$ which vanishes on $S$, we need by proposition $4$ that $m+ld \geq k(h-1)$. Similarly, to ensure that $f(\vec{x}) \in F[x_{1},\dots,x_{k}]$ of degree at most $d$ having at least $t$ agreements on $S$ satisfies that $y - f(x) \mid P(\vec{x},y)$, we require that $(m+ld) < \frac{t}{h^{k-1}}$. Hence, we want $m,l \in \N$ such that $\frac{t}{h^{k-1}} > m+ld > k(h-1)$. Considering simpler case of setting $m = 0$, and finding the appropriate $l$, we obtain that such an $l$ exists precisely when $I = \pren{\frac{k(h-1)}{d},\frac{t}{dh^{k-1}}}$ contains a positive integer. Assuming that the given parameters are such that $l$ exists, by proposition 4, we obtain a non-zero polynomial vanishing on $S$. and by proposition $5$, we have that a polynomial $f \in F[\vec{x}]$ of degree  at most $d$ having at least $t$ agreements on $S$ will have $y-f$ divide $P(\vec{x},y)$. Hence, Algorithm 2 will return the desired polynomials. 
\end{proof}

\section{References}

\begin{itemize}
    \item M. Sudan, Decoding of Reed-Solomon Codes Beyond the Error Correction Bound. http://people.csail.mit.edu/madhu/papers/1996/reeds-journ.pdf
\end{itemize}

\end{document}